\documentclass[11pt,a4paper]{article}
\usepackage[margin=22mm]{geometry}
\usepackage{amsmath,amssymb,amsthm,booktabs,array,microtype,enumitem,xcolor,graphicx,tikz}
\usepackage[round,authoryear]{natbib}
\usepackage[colorlinks=true,linkcolor=blue!45!black,citecolor=blue!45!black,urlcolor=blue!45!black]{hyperref}
\setlist{itemsep=2pt,topsep=4pt}
\newtheorem{proposition}{Proposition}
\newtheorem{corollary}{Corollary}
\newcommand{\doi}[1]{\href{https://doi.org/#1}{doi:\nolinkurl{#1}}}
\makeatletter
\newenvironment{fixedtable}{\par\addvspace{10pt}\noindent\begin{minipage}{\linewidth}\def\@captype{table}}{\end{minipage}\par\addvspace{10pt}}
\newenvironment{fixedfigure}{\par\addvspace{8pt}\noindent\begin{minipage}{\linewidth}\def\@captype{figure}}{\end{minipage}\par\addvspace{8pt}}
\makeatother
\hypersetup{pdftitle={Scenario Matching for Credit Dependence Stress Tests},pdfauthor={Naoufal NIFA},pdfsubject={Preprint, version 1.0.0; CC BY-NC 4.0},pdfkeywords={scenario coupling, temporal correlation, credit stress testing, assignment, wrong-way risk, empirical tail bounds}}
\begin{document}
\begin{center}
{\LARGE\bfseries Scenario Matching for Credit\\Dependence Stress Tests}\\[8pt]
{\large Temporal correlation constraints and empirical risk bounds}\\[7pt]
{\normalsize Naoufal NIFA}\\[2pt]
{\small Unaffiliated researcher \quad\textbar\quad \href{https://orcid.org/0009-0006-6983-6537}{ORCID: 0009-0006-6983-6537}}\\[5pt]
15 October 2026

\end{center}
\begin{abstract}
We couple two sets of simulated market and credit paths by minimizing absolute temporal-correlation error relative to a prescribed target. Each market path receives exactly one credit path, preserving both empirical distributions of complete scenarios. A perfect-matching test identifies whether a uniform correlation-error tolerance is feasible and determines its smallest attainable value. Further assignments give exact empirical mean-loss and exceedance-probability bounds; quantile bounds follow by binary search over eligible loss thresholds. A synthetic bond-forward credit stress test illustrates the method. The experiments report matching precision, empirical score tails and exposure-grid error, and measure matching and mean-loss computation. The method is independent of the scenario generator and provides a reproducible way to examine dependence assumptions in existing risk engines. Tail measures concern default-probability-weighted credit scores.
\end{abstract}

\section{Motivation and relationship to existing work}
Market exposures and counterparty credit quality are often simulated by separate engines. Their marginal models may be useful even when their dependence is uncertain. We retain both complete sets of scenario paths and ask which pairings approximate a prescribed \emph{within-path} correlation, and how much credit risk remains undetermined among admissible pairings.

The optimization machinery is established. Hungarian assignment solves additive one-to-one matching problems \citep{kuhn1955,burkard2012}, while marginal-preserving correlation construction has a long history \citep{iman1982}. More directly, \citet{glasserman2018} optimize market--credit coupling with fixed marginal simulations to bound wrong-way risk, including relative-entropy restrictions; \citet{memartoluie2017} also study counterparty risk bounds. The contribution considered here is a temporal-correlation eligibility criterion, its feasibility for finite scenario sets, and its financial interpretation. Neither fixed-marginal coupling nor general loss bounds are claimed as new.

The assignment is a discrete optimal-transport problem. With equal weights and only eligible-edge restrictions, unregularized linear transport admits a permutation optimum \citep{peyre2019}. Entropic regularization changes the objective and generally assigns probability mass to several partners, enabling Sinkhorn scaling \citep{cuturi2013}. Both constructions can preserve intact paths and their weighted empirical marginals. Assignment produces exactly $n$ pairs of weight $1/n$, which can enter an equal-weight risk engine directly. This is the operational distinction; a numerical comparison with Sinkhorn is outside the experiment.

\section{Matching complete scenario blocks}
Let $C_1,\ldots,C_n$ be equally weighted market scenarios and $Z_1,\ldots,Z_n$ equally weighted credit scenarios. Each block contains all jointly generated variables that must remain together. Scalar features $u_{ik}=g(C_i,t_k)$ and $v_{jk}=h(Z_j,t_k)$ are observed on the same $L$ dates and must have nonzero temporal variance. Define
\begin{equation}\label{eq:R}
R_{ij}=\frac{\sum_{k=1}^{L}(u_{ik}-\bar u_i)(v_{jk}-\bar v_j)}
{\sqrt{\sum_k(u_{ik}-\bar u_i)^2\sum_k(v_{jk}-\bar v_j)^2}},
\qquad d_{ij}=|R_{ij}-\rho_\star|,
\end{equation}
where $\rho_\star$ is the historical level-correlation estimate or prescribed stress target. This cross-block array need not be symmetric or positive semidefinite. If $\mathcal E$ contains eligible pairs, choose
\begin{equation}\label{eq:assignment}
\pi_\star\in\arg\min_{\pi\in\Pi_{\mathcal E}}\frac1n\sum_i d_{i,\pi(i)},
\end{equation}
where $\Pi_{\mathcal E}$ is the set of permutations using eligible edges. Reorder the \emph{entire} destination block by the same permutation. The original unconstrained method takes every pair as eligible; financial compatibility conditions may remove edges.

Figure~\ref{fig:method} separates target fitting from risk optimization. Every source receives one destination, and every destination is used once. Choosing each source's closest destination independently can reuse the same destination path and therefore fail to produce a valid assignment.

\begin{fixedfigure}
\centering
\begin{tikzpicture}[
  x=1mm,y=1mm,
  font=\fontsize{9}{11.2}\selectfont,
  every node/.style={outer sep=0pt},
  stage/.style={draw=blue!35!black!40,fill=blue!3,rounded corners=1.5mm,
    line width=.55pt,align=left,inner xsep=2.5mm,inner ysep=3mm},
  riskstage/.style={stage,draw=green!30!black!45,fill=green!3},
  flow/.style={->,draw=black!55,line width=.7pt,>=stealth}
]
\node[stage,text width=37mm,minimum height=58mm] (scenarios) at (21,36) {
  \textbf{1\quad Complete paths}\\[2.4mm]
  \textbf{Market block $C_i$}\\
  Curves and discounting\\
  Signed value $M_i(t)$\\[2.1mm]
  \textbf{Credit block $Z_j$}\\
  Hazard $\lambda_j(t)$\\[2.1mm]
  Use aligned dates; retain\\
  each complete block.
};
\node[stage,text width=42mm,minimum height=58mm] (eligible) at (71.5,36) {
  \textbf{2\quad Define eligible pairs}\\[2.4mm]
  $R_{ij}=\operatorname{Corr}_{\rm time}(M_i,\lambda_j)$\\[1.2mm]
  $d_{ij}=|R_{ij}-\rho_\star|$\\[2.1mm]
  Keep compatible pairs\\
  with $d_{ij}\leq\varepsilon$.\\[2.1mm]
  A perfect matching\\
  must exist. Otherwise,\\
  report infeasibility.
};
\node[stage,text width=55mm,minimum height=26mm] (fit) at (130,53) {
  \textbf{3a\quad Fit correlation}\\[2mm]
  Minimize $\displaystyle\frac1n\sum_i d_{i,\pi(i)}$.\\[1.3mm]
  One partner per path;\\
  every path used exactly once.
};
\node[riskstage,text width=55mm,minimum height=30mm] (risk) at (130,16) {
  \textbf{3b\quad Bound mean and tail}\\[2mm]
  Minimize/maximize $\displaystyle\frac1n\sum_i c_{i,\pi(i)}$.\\[1.5mm]
  Mean: $c_{ij}=L_{ij}$.\\[1.1mm]
  Tail: $c_{ij}=\mathbf 1\{L_{ij}>\theta\}$.\\[1.1mm]
  Use the same eligible graph.
};
\draw[flow] (scenarios.east) -- (eligible.west);
\draw[draw=black!55,line width=.7pt] (eligible.east) -- (97.5,36);
\draw[flow] (97.5,36) |- (fit.west);
\draw[flow] (97.5,36) |- (risk.west);
\end{tikzpicture}%

\caption{Complete-path assignment. Target $\rho_\star$ concerns temporal correlation within each pair; tolerance $\varepsilon$ defines eligible pairs. Correlation fitting and loss or exceedance optimization share this graph but use different costs. Reassignment leaves simulated path values unchanged.}
\label{fig:method}
\end{fixedfigure}

\begin{proposition}[Preservation, feasibility and loss bounds]\label{prop:main}
If $\Pi_{\mathcal E}$ is nonempty, \eqref{eq:assignment} is a linear assignment and has a global optimum. Every such permutation preserves the two empirical marginal distributions on complete blocks. A uniform error tolerance $\varepsilon$ is attainable exactly when the graph
\begin{equation}\label{eq:graph}
\mathcal E_\varepsilon=\{(i,j)\in\mathcal E:d_{ij}\leq\varepsilon\}
\end{equation}
has a perfect matching. For any finite fixed pairwise monetary loss $L_{ij}$ and nonempty $\Pi_{\mathcal E_\varepsilon}$, the attainable expectation endpoints are
\begin{equation}\label{eq:bounds}
\underline\ell=\min_{\pi\in\Pi_{\mathcal E_\varepsilon}}\frac1n\sum_i L_{i,\pi(i)},
\qquad
\overline\ell=\max_{\pi\in\Pi_{\mathcal E_\varepsilon}}\frac1n\sum_i L_{i,\pi(i)}.
\end{equation}
They are obtained by ordinary assignments with edge costs $L$ and $-L$.
\end{proposition}
\begin{proof}
All objectives are sums of selected edge costs under bijection constraints. Relabeling blocks leaves their empirical multisets unchanged. A feasible permutation is precisely a perfect matching of the eligible graph. The finite set of feasible permutations attains both loss extrema; changing the sign of costs converts maximization to minimization.
\end{proof}

\begin{corollary}[Empirical tail and quantile bounds]\label{cor:tail}
Assume $\Pi_{\mathcal E_\varepsilon}\ne\varnothing$. At any threshold $\theta$, set
$B_{ij}(\theta)=\mathbf 1\{L_{ij}>\theta\}$ and
$p_\pi(\theta)=n^{-1}\sum_i B_{i,\pi(i)}(\theta)$.
The exact attainable empirical exceedance bounds are
\begin{equation}\label{eq:tail}
\underline p(\theta)=\min_{\pi\in\Pi_{\mathcal E_\varepsilon}}p_\pi(\theta),
\qquad
\overline p(\theta)=\max_{\pi\in\Pi_{\mathcal E_\varepsilon}}p_\pi(\theta).
\end{equation}
They are linear assignments with costs $B(\theta)$ and $-B(\theta)$.
For $\alpha\in(0,1]$, let $k=\lceil n\alpha\rceil$ and define $q_\alpha(\pi)$ as the $k$th smallest selected loss, including ties. With
$\mathcal S=\{L_{ij}:(i,j)\in\mathcal E_\varepsilon\}$,
\begin{align}\label{eq:quantile}
\min_\pi q_\alpha(\pi)&=\min\{\theta\in\mathcal S:\underline p(\theta)\leq(n-k)/n\},\\
\max_\pi q_\alpha(\pi)&=\min\{\theta\in\mathcal S:\overline p(\theta)\leq(n-k)/n\}.\notag
\end{align}
\end{corollary}
\begin{proof}
An indicator is a fixed edge cost, so Proposition~\ref{prop:main} applies. The $k$th selected order statistic is at most $\theta$ exactly when at most $n-k$ selected losses strictly exceed $\theta$. For the minimum quantile this condition must hold for at least one feasible permutation; for the maximum it must hold for all of them. This gives the two inversions, with extrema attained on the finite loss set.
\end{proof}

The tail bounds are pointwise: their optimizing permutations can change with $\theta$. Consequently, an entire envelope need not be the tail curve of one feasible coupling. Quantile bounds require threshold searches, not merely minimization or maximization of mean loss. These results specialize standard assignment to indicator costs; they are not claimed as a new general tail-risk theorem.

Both functions in \eqref{eq:tail} are nonincreasing step functions of $\theta$. After sorting the distinct values in $\mathcal S$, binary search locates each quantile endpoint using at most $\lceil\log_2|\mathcal S|\rceil$ assignment queries: at most 16 when $n=256$, since $|\mathcal S|\leq n^2$. This counts search queries only; the implementation makes up to two additional calls to verify the endpoint and its predecessor and retain an attaining permutation.

A low average correlation error does not guarantee that every pair meets tolerance. By Hall's theorem, feasibility requires every subset of source scenarios to have at least as many destination neighbors \citep{hall1935}. The minimum feasible tolerance is the bottleneck value $\varepsilon_{\min}=\min_{\pi\in\Pi_{\mathcal E}}\max_i d_{i,\pi(i)}$. An infeasible request should be reported, not silently relaxed. Dense construction costs $O(n^2L)$ arithmetic and $O(n^2)$ storage; efficient dense Hungarian implementations take $O(n^3)$ worst-case time.

Each edge loss uses a fixed policy and must not depend on other selected pairs. A uniformly random permutation has expected loss $n^{-2}\sum_{ij}L_{ij}$, an independence reference that need not satisfy \eqref{eq:graph}. An additive loss $a_i+b_j$ has invariant expectation: changing expected loss requires an interaction between blocks.

\subsection{Interpretation and model scope}\label{sec:scope}
Four distinctions govern the method and the numerical results.
\begin{enumerate}[leftmargin=*,label=\arabic*.]
\item \textbf{Targeted dependence.} The feature is correlation across dates within a pair, not across scenarios at a future date. Level matching addresses the coincidence of high market values and high credit-risk states; first-difference matching addresses co-movement of changes. Increments are not generally innovations without accounting for conditional drift. Neither statistic uniquely determines monetary loss.
\item \textbf{Finite-window interpretation.} Window length, frequency and the inclusion of $t_0$ must match the target-estimation convention. A deterministic starting observation changes each path's sample mean and centered cross-products; Pearson correlation remains centered on the sample mean, not the initial value. Excluding $t_0$ does not remove initial conditions from the dynamics. Persistent levels and short windows limit statistical inference. Here targets are prescribed stresses and the numerical generators are uncalibrated; no historical forecasting evidence is claimed.
\item \textbf{Meaning of risk bounds.} Ranges are attainable extrema over pairings of fixed scenario sets, not sampling confidence intervals. In the application, $L_{ij}$ is already a default-probability-weighted credit score. Its empirical tail and quantiles describe that score, not realized jump-to-default cash losses or regulatory portfolio VaR. A pointwise tail envelope need not arise from one coupling at every threshold.
\item \textbf{Joint-model consistency.} Keeping curves, factors, discounting and market values together preserves the market block. Offline full-path pairing does not establish causal consistency, a joint risk-neutral law or a calibrated CVA price. Policies cannot use the pairing as foreknowledge. Structural dependence of credit dynamics on market states requires conditional modeling or additional restrictions.
\end{enumerate}

Only aligned features and pairwise costs are required. Gaussian models such as G2++, square-root models such as CIR2++ \citep{brigo2006}, affine arbitrage-free Nelson--Siegel \citep{christensen2011}, and HJM forward-curve models \citep{hjm1992} are alternative input generators. The numerical example uses a Gaussian model for tractable bond values and exact sampling of discount integrals. A historically supported target band can replace a point target through the cost $\operatorname{dist}(R_{ij},[\rho_-,\rho_+])$, subject to validation.

\section{Credit dependence stress test}
Wrong-way risk arises when positive exposure is high as counterparty credit deteriorates. We impose temporal-correlation restrictions on fixed market and credit scenario sets, retaining exposure and discounting in the same market block. Let $0=t_0<\cdots<t_m=H$ with $\Delta=H/m$, bond prices $P_i(t,T)$ and realized discount factors $D_i(t)$.

Consider a long forward on a zero-coupon bond, with notional $N$, delivery $H$ and bond maturity $T_B>H$. At strike $K=P(0,T_B)/P(0,H)$, signed market value and positive exposure are
\begin{equation}\label{eq:mtm}
M_i(t)=N[P_i(t,T_B)-K P_i(t,H)],\qquad E_i(t)=\max\{M_i(t),0\}.
\end{equation}
The credit block supplies hazard $\lambda_j(t_k)>0$, held constant during the following interval. Survival and interval-default weights are
\begin{equation}\label{eq:survival}
S_j(t_k)=\exp\left[-\Delta\sum_{\ell=0}^{k-1}\lambda_j(t_\ell)\right],\qquad
q_{jk}=S_j(t_k)\bigl[1-e^{-\lambda_j(t_k)\Delta}\bigr].
\end{equation}
For loss-given-default $\gamma$, the pairwise discounted credit score is
\begin{equation}\label{eq:credit}
L_{ij}=\gamma\sum_{k=0}^{m-1}D_i(t_{k+1})E_i(t_{k+1})q_{jk}.
\end{equation}
The weights telescope to $1-S_j(H)$. Exposure and discounting use an interval-end approximation; the numerical grid study below examines its precision.

\paragraph{Targeted dependence and measured loss.}
The feature is temporal correlation between \emph{signed} value $M_i(t_k)$ and hazard $\lambda_j(t_k)$ at period starts. A positive target aligns above-average values with above-average hazards within each pair. Signed values remain informative even when positive exposure is identically zero. The monetary score instead uses positive period-end exposure, default weights and discounting. For example, EUR\,1 million positive exposure with $60\%$ loss-given-default and $2.5\%$ annual hazard contributes approximately EUR\,1,249 over one month before prior-survival weighting and discounting. Pearson correlation removes levels and scales, whereas this score retains them; the target therefore does not determine a unique financial outcome. Collateral, netting, bilateral default and credit-dependent closeout are excluded from this example.

\section{Numerical design and results}
The curve generator has a flat initial forward rate of $2.5\%$. Two zero-start Gaussian Ornstein--Uhlenbeck factors have mean reversions $(0.1,0.5)$, volatilities $(0.005,0.01)$ and Brownian correlation $-0.5$; a deterministic shift fits the initial curve. Factors and the integrated short rate are sampled jointly from their exact Gaussian transition. Bonds use the generator's conditional valuation formula, and $D_i(t)=\exp[-\int_0^t r_i(u)\,du]$.

An independently generated log driver satisfies $dw=-0.8w\,dt+0.6\,dB$, starting at zero, with exact transitions. Set
\begin{equation}\label{eq:driver}
\lambda_j(t)=\lambda_0\exp\left[w_j(t)-\tfrac12\operatorname{Var}w_j(t)\right],
\qquad \lambda_0=0.025.
\end{equation}
The variance correction fixes the unconditional hazard mean. We use $N=\text{EUR }100$ million, $\gamma=0.6$, $H=1$ year, $T_B=10$ years and twelve monthly periods. Features use the twelve period starts, including $t_0$; losses use the conventions in \eqref{eq:credit}.

Ten independent batches each contain $n=256$ market paths and $n=256$ credit paths. The same scenario sets are used for targets $\rho_\star\in\{-0.5,0.0,0.5,0.8\}$, with tolerance $\varepsilon=0.2$ fixed throughout. Feasibility is checked separately before a target is evaluated. This separates changing dependence assumptions from changing either marginal model. The level and $t_0$ conventions are fixed; no sensitivity comparison of those conventions is included.

\subsection{Target sensitivity and monetary consequences}
Table~\ref{tab:applicationresults} reports ten-batch averages. Achieved correlation and mean absolute error describe the minimum-error assignment. Its credit score is compared with the minimum and maximum expected scores over the same eligible graph.

\begin{fixedtable}
\centering\small
\setlength{\tabcolsep}{5pt}
\caption{Target sweep on the same ten scenario sets, $n=256$, tolerance $0.2$. Monetary values are EUR thousands. Dashes denote infeasibility; all reported monetary averages use ten feasible batches.}
\label{tab:applicationresults}
\begin{tabular}{rrrrrr}
\toprule
Target & Feasible & Achieved & Mean error & Selected score & Permitted range\\
\midrule
$-0.5$ & 10/10 & $-0.4996$ & 0.0075 & 6.14 & [5.07, 7.84]\\
$0.0$ & 10/10 & $-0.0002$ & 0.0154 & 7.46 & [5.81, 9.43]\\
$0.5$ & 10/10 & $0.4996$ & 0.0061 & 8.90 & [7.06, 10.48]\\
$0.8$ & 0/10 & --- & --- & --- & ---\\
\bottomrule
\end{tabular}
\end{fixedtable}

Targets $-0.5$, $0.0$ and $0.5$ are feasible in every batch. The selected mean score rises from EUR\,6,142 to EUR\,7,461 and EUR\,8,899, respectively. The increase from $-0.5$ to $0.5$ is $44.9\%$, using exactly the same marginal scenario sets. The $0.0$-to-$0.5$ increase is EUR\,1,438 ($19.28\%$), with paired Monte Carlo standard error EUR\,25. The random-pairing reference is EUR\,7,429; its closeness to the target-zero result does not make zero temporal correlation equivalent to independence.

At target $0.8$, no batch admits a complete matching within $0.2$. The computed minimum tolerances range approximately from $0.2131$ to $0.3765$. Thus a more positive target cannot simply be imposed on these fixed scenario sets under the requested precision. No monetary result is substituted for the failed cases.

\begin{fixedfigure}
\centering
\includegraphics[width=\linewidth]{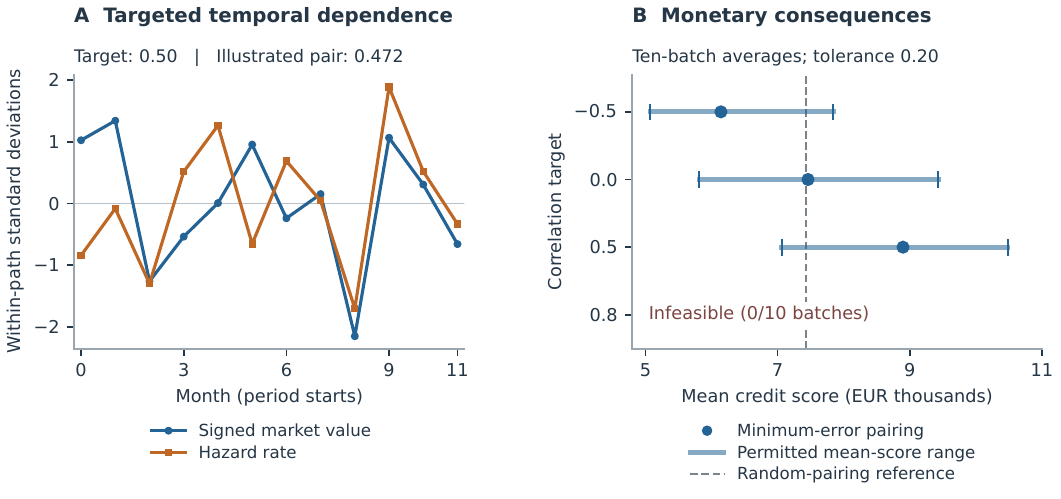}
\caption{Targeted dependence and monetary sensitivity. Left: separately standardized signed market value and hazard for the first market path in seed 20261014 and its partner selected for target $0.5$ (credit path 92). Zero is each path's own mean. Right: dots give average minimum-error scores; horizontal segments join average batch-specific monetary extrema at tolerance $0.2$. The dashed line is random pairing. The row for target $0.8$ records infeasibility.}
\label{fig:credit}
\end{fixedfigure}

The illustrated pair has correlation $0.472$; the ten-batch average at target $0.5$ is $0.4996$. A standardized signed value above zero need not be positive in euros. The monetary calculation applies the positive-part rule to actual values, which is why EUR\,100 million notional produces scores in thousands rather than millions. At target $0.5$, the permitted mean-score range is EUR\,7,061--10,483, with width EUR\,3,422, or $38.5\%$ of the selected score. Near-target correlation leaves economically material dependence uncertainty.

\subsection{Tail probabilities and scenario quantiles}
To illustrate Corollary~\ref{cor:tail}, fix an illustrative score threshold of EUR\,20,000 before calculation. The same twenty target-$0.0$/$0.5$ cases supply exact exceedance bounds and empirical 95th/99th quantile bounds. For 256 equally weighted paths, these quantiles are the 244th and 254th selected order statistics, without interpolation.

\begin{fixedtable}
\centering\small
\caption{Empirical score tails at tolerance $0.2$: averages of ten batch-specific statistics and exact endpoints. Quantiles are EUR thousands; exceedance probabilities are percentages.}
\label{tab:tail}
\begin{tabular}{llrrr}
\toprule
Target & Statistic & Selected & Minimum & Maximum\\
\midrule
$0.0$ & $\Pr(L>20{,}000)$ (\%) & 12.38 & 4.26 & 22.11\\
$0.0$ & 95th percentile (EUR k) & 28.06 & 19.33 & 40.40\\
$0.0$ & 99th percentile (EUR k) & 41.88 & 27.76 & 59.07\\
\addlinespace
$0.5$ & $\Pr(L>20{,}000)$ (\%) & 16.45 & 8.28 & 25.47\\
$0.5$ & 95th percentile (EUR k) & 33.57 & 23.72 & 49.31\\
$0.5$ & 99th percentile (EUR k) & 50.01 & 34.13 & 73.54\\
\bottomrule
\end{tabular}
\end{fixedtable}

Under minimum-error pairing, the fraction of scores above EUR\,20,000 rises from $12.38\%$ to $16.45\%$. At target $0.5$, however, admissible pairings allow average exceedance probabilities from $8.28\%$ to $25.47\%$. Likewise, the selected 99th percentile is approximately EUR\,50,010, while the average attainable endpoint range is EUR\,34,133--73,541. This translates dependence uncertainty into a score-budget exceedance question rather than only a change in mean cost.

Each table entry averages ten batch-specific statistics; quantiles are not computed after pooling batches. Endpoint optimizers can differ across statistics, so the best or worst values in one row need not coexist with those in another. Witness permutations and threshold certificates verify attainability and the quantile inversion. At this sample size the 99th percentile is determined by only the top few observations, making the finite-scenario interpretation in Section~\ref{sec:scope} particularly relevant.

\subsection{Tolerance feasibility and correlation precision}
At target $0.5$, mean absolute error is $0.0061$, but the largest selected-pair error over all ten batches is $0.0861$. At target $0.0$, the corresponding maximum is $0.1983$. Average fit and a uniform per-pair guarantee are distinct criteria.

For target $0.5$, tightening $\varepsilon$ from $0.2$ to $0.1$ remains feasible in every batch and reduces average monetary-range width from EUR\,3,422 to EUR\,2,833, a $17.2\%$ reduction. Minimum-error pairings are unchanged because they already meet $0.1$. At $\varepsilon=0.05$, only eight batches are feasible. The two failing seeds have computed bottleneck values
\begin{equation}\label{eq:failedbottlenecks}
\varepsilon_{\min}\approx0.0765\quad(20261208),\qquad
\varepsilon_{\min}\approx0.0717\quad(20261887).
\end{equation}
These values minimize the \emph{largest} error, rather than reporting the worst error of a minimum-average-error assignment. Full numerical precision and witness permutations are supplied in the data. At $0.05$, the first failure has a Hall-deficient set of 255 market paths with only 254 eligible credit paths; the second contains a market path with no eligible partner. Increasing the tolerance addresses a concrete matching obstruction.

The monetary bounds use the entire permitted band. At target $0.5$ and tolerance $0.2$, the minimum- and maximum-loss pairings have average correlations $0.4053$ and $0.5386$, compared with $0.4996$ for the minimum-error pairing. The remaining monetary range therefore holds eligibility fixed, not achieved mean correlation exactly fixed.

\subsection{Precision of the exposure calculation}
Exact factor transitions and discount integrals do not remove interval-end exposure approximation. A separate study generates ten new nested-path batches, fixes the monthly pairings and hazards, and subdivides the exposure/discount evaluation grid. Default probabilities remain unchanged across grids.

\begin{fixedtable}
\centering\small
\caption{Exposure-grid sensitivity: ten new nested-path batches, $n=256$, target $0.5$. Monthly pairings and hazards are fixed. Gaps use the mean 192-date score as denominator.}
\label{tab:creditgrid}
\begin{tabular}{rrr}
\toprule
Evaluation dates & Mean score (EUR) & Gap to 192 dates (\%)\\
\midrule
12 & 9,584 & 4.95\\
24 & 9,347 & 2.35\\
48 & 9,224 & 1.01\\
96 & 9,163 & 0.34\\
192 & 9,132 & 0.00\\
\bottomrule
\end{tabular}
\end{fixedtable}

Moving from 12 to 192 evaluation dates lowers the target-$0.5$ score by EUR\,452, or $4.95\%$ of the finer-grid value; the paired Monte Carlo standard error is EUR\,18. The 96-to-192 gap is EUR\,31, or $0.34\%$. The 192-date value is a numerical reference. Comparisons within Table~\ref{tab:creditgrid} isolate grid sensitivity; the main monthly table uses different draws. Correlation precision cannot substitute for resolving the financial loss calculation in time.

Independent scalar cash-flow checks verify the credit formulas and saved permutations. Exhaustive small assignments separately validate the theoretical extension. Arithmetic agreement, Monte Carlo variability and time-grid sensitivity are reported as different aspects of numerical precision.

\subsection{Performance and scalability}
The credit benchmark uses five seeds at each size, twelve monthly observations, target $0.5$ and tolerance $0.2$. It was measured on an Apple M3 MacBook Air with 8 GB RAM, Python 3.12.14, NumPy 2.5.4 and SciPy 1.18.1; numerical backend thread-limit variables were set to one. Warm-up is excluded. Table~\ref{tab:performance} times path generation, pairwise loss construction, correlation matching, feasibility checks and \emph{two mean-loss bounds}. Imports, disk output and validation are excluded. New bottleneck, tail and quantile searches are not included in these timings.

\begin{fixedtable}
\centering\small
\caption{Credit matching and mean-loss bounds: five seeds per size. Time is the median workflow duration; error is mean absolute correlation deviation averaged across batches. Storage covers three dense matrices only.}
\label{tab:performance}
\begin{tabular}{rrrr}
\toprule
Paths & Workflow (s) & Mean error & Matrix storage (MiB)\\
\midrule
128 & 0.006 & 0.01285 & 0.375\\
256 & 0.018 & 0.00617 & 1.5\\
512 & 0.084 & 0.00314 & 6\\
1,024 & 0.502 & 0.00158 & 24\\
2,048 & 4.529 & 0.00079 & 96\\
\bottomrule
\end{tabular}
\end{fixedtable}

The matching-and-mean-loss workflow takes about 18 milliseconds at 256 scenarios and 4.53 seconds at 2,048; the latter ranges from 3.88 to 6.47 seconds across seeds. At 2,048, the complete monetary-bound routine has median time 3.92 seconds, compared with 0.46 seconds for the correlation-matching routine. The bounds call performs two assignments and includes its own preparation and feasibility work. The recorded clocks identify the expensive routine but do not measure dual degeneracy or augmenting-path lengths.

The unmasked credit matrix has the exact factorization
\begin{equation}\label{eq:factorization}
L=BQ^\top,\qquad B_{ik}=\gamma D_i(t_{k+1})E_i(t_{k+1}),\quad Q_{jk}=q_{jk},
\qquad\operatorname{rank}(L)\leq m=12.
\end{equation}
Likewise, centering twelve-date features gives $\operatorname{rank}(R)\leq11$; taking entrywise absolute errors can increase rank but does not guarantee full rank. These algebraic bounds do not establish solver runtime. Forbidden edges are represented by infinite costs, and the solver operates on that restricted input. The implementation uses SciPy's modified Jonker--Volgenant algorithm for the same assignment objective solved by Hungarian methods.

All 25 credit benchmark cases are feasible. Mean matching error falls from $0.01285$ at 128 paths to $0.00079$ at 2,048, with a substantial increase in computation. Three dense double-precision matrices require $24n^2$ bytes: 96 MiB at 2,048, plus 4 MiB for a Boolean mask and additional working storage. At 10,000 paths those matrices alone would require 2.4 GB; that size was not timed.

The factors in \eqref{eq:factorization} need only $2nm$ numbers: 0.375 MiB instead of 32 MiB for one dense matrix at 2,048 paths. Exploiting this during optimization requires an implicit-cost solver; the measured implementation still materializes its arrays. Reusing matrices across targets avoids repeated construction, while candidate screening may reduce memory at the cost of excluding feasible or optimal pairings. Computing both exceedance endpoints at a fixed threshold requires two assignments. For each quantile endpoint, binary search instead uses one assignment per query, as described after Corollary~\ref{cor:tail}; the implemented search and verification require at most 18 calls at $n=256$. Sorting the $E=|\mathcal E_\varepsilon|$ eligible losses costs $O(E\log E)$ separately.

\section{Implementation and practical use}
The procedure returns a block permutation, its correlation errors, the minimum feasible tolerance and risk bounds for a fixed policy. Applying the loss-bound routine to $\mathbf 1\{L>\theta\}$ requires no change to the assignment algorithm. Exhaustive small examples with ties, negative losses and infeasible graphs verify the tail and quantile extension.

For deployment, specify the exposure and closeout calculation before pairing, estimate marginal dynamics and dependence targets on training data, select tolerances on validation data, and assess outcomes on held-out periods against a credible joint or conditional model. Collateral, netting and multiple counterparties require corresponding cash-flow and dependence constraints. The useful output is a transparent account of which risk conclusions are fixed by the scenario sets and which still depend on the chosen coupling.

\section*{Reproducibility}
The package contains the scenario sets, seeds, permutations, financial formulas, target and tail studies, independent audits, timing records, nested-grid data and vector figures. The archived version of this work is available on Zenodo at
\doi{10.5281/zenodo.23405059}. All reported financial experiments use the disclosed synthetic inputs. The accompanying reproducibility archive provides all inputs needed to reproduce the reported results. This paper and its original non-code supporting materials are licensed under \href{https://creativecommons.org/licenses/by-nc/4.0/}{CC BY-NC 4.0}; computational code is distributed under the PolyForm Noncommercial License 1.0.0.

\section*{Declaration of AI assistance for code development}

AI tools assisted with developing the code used to generate the reported results. The author takes full responsibility for the methodology, code, numerical results, interpretations, and final content.

\begingroup\small\setlength{\bibsep}{0pt}

\endgroup
\end{document}